\documentclass[12pt,a4paper,final]{article}
\usepackage[utf8]{inputenc}
\usepackage[english]{babel}
\usepackage[T1]{fontenc}
\usepackage{graphicx}
\usepackage{array}
\usepackage{listings}
\usepackage{float}
\usepackage{dsfont}
\usepackage{gensymb}
\usepackage{amsthm}
\usepackage{amsmath,amsfonts,amssymb}
\usepackage{stmaryrd}
\usepackage{cancel}
\usepackage{titlesec}
\usepackage{bm}
\usepackage{hyperref}
\usepackage{appendix}
\usepackage[colorinlistoftodos]{todonotes}
\usepackage[nottoc]{tocbibind}

\usepackage{fancyvrb}
\RecustomVerbatimCommand{\VerbatimInput}{VerbatimInput}%
{fontsize=\footnotesize,
 frame=lines,  
 commentchar=*        
}

\DeclareUnicodeCharacter{00A0}{ }
\author{Luna Matsuo \& Meziane Privat}
\title{Stage}

\graphicspath{{Images/}}
\usepackage{biblatex}
\titleformat{\paragraph}
{\normalfont\normalsize\bfseries}{\theparagraph}{1em}{}
\titlespacing*{\paragraph}{0pt}{3.25ex plus 1ex minus .2ex}{1.5ex plus .2ex}

\newtheorem{definition}{Définition}[section]

\newtheorem{prop}{Proposition}[section]

\newenvironment{changemargin}[2]{\begin{list}{}{%
\setlength{\topsep}{0pt}%
\setlength{\leftmargin}{0pt}%
\setlength{\rightmargin}{0pt}%
\setlength{\listparindent}{\parindent}%
\setlength{\itemindent}{\parindent}%
\setlength{\parsep}{0pt plus 1pt}%
\addtolength{\leftmargin}{#1}%
\addtolength{\rightmargin}{#2}%
}\item }{\end{list}}
\begin{document}
\hypersetup{pdfborder=0 0 0}
\begin{center}
{\LARGE \textbf{Infinite horizon for symetric strategy population game}}
\\~~\\ \textbf{Meziane Privat}
\end{center}~~\\
\begin{changemargin}{1.5cm}{1.5cm}
\begin{small}
 \textbf{Abstract : } To predict the behavior of a population game when time becomes very long, the process that characterizes the evolution of our game dynamics must be reversible. Known games satisfying this are 2 strategy games as well as potential games with an exponential protocol. We will try to extend the study of infinite horizons for what are called symetric strategy games.\\
 
 \textbf{Key words :} Population Games,  Stochastic Evolutionary Models,  Infinite-Horizon.\\
 
 \textbf{AMS subject classification :} 91A15,  91A22,  92D25
 \end{small} 
 \end{changemargin}~~\\
\tableofcontents
\newpage

\addcontentsline{toc}{section}{Introduction}
\section*{Introduction}
In this introduction we define the notations that we use for the rest of the article, we use the same notations as in \cite{San10}.\\

A population game can be thought of as game theory where agents are populations. From a purely formal point of view, a population game is defined by an arbitrary collection of real-valued continuous functions on an appropriate domain.\\
Population games model many strategic interactions with the following properties:
\begin{itemize}
\item[(i)] The number of agents is important.
\item[(ii)] Individual agents are small. The behavior of one agent has little or no effect on the earnings of other agents.
\item[(iii)] Agents interact anonymously. The gains of each agent depend only on the behavior of the adversaries through the distribution of their choices.
\end{itemize}
Although these three properties are fundamental, two additional restrictions of a more technical nature must be taken into account. For the distribution of opponents' choices mentioned in property (iii) to exist, there must be collections of agents that share the same set of strategies. This structure and many others are ensured by the following property:
\begin{itemize}
\item[(iv)] The number of roles is finite. Each agent is a member of a finite number of populations. Members of a population choose from the same set of strategies, and their payoffs are identical functions of their own behavior and the distribution of the behavior of opponents.
\end{itemize}
The last property ensures that very small changes in the overall behavior do not lead to large changes in the payoffs of the strategies:
\begin{itemize}
\item[(v)] Gains are continuous. The dependence of each agent's payoffs on the distribution of opponents' choices is continuous.
\end{itemize}

Let $\mathcal{P} = \lbrace 1, ... ,\mathfrak{p}\rbrace$ be a society composed of $\mathfrak{p} \geq 1$ populations of agents. An agent in the population $p$ forms a continuous mass $m^p > 0$. (Thus, $\mathfrak{p}$ is the number of populations, while $p$ is an arbitrary population.)\\

\begin{definition} The set of strategies available to agents of population $p$ is denoted $S^p = \lbrace 1, ... , n^p\rbrace$. The total number of pure strategies in all populations is denoted $n=\sum\limits_{p\in\mathcal{P}}n^p$.
\end{definition}

\begin{definition} The set of population states (or strategy distributions) for population $p$ is $X^p = \lbrace x^p\in\mathbb{R}^{n^p}_+ : \sum\limits_{i\in S^p}x_i^p = m^p\rbrace$. The scalar $x_i^p\in\mathbb{R}_+$ represents the mass of players in the population $p$ choosing the strategy $i \in S^p$.
\end{definition}

\begin{definition} A payoff function $F : X \rightarrow\mathbb{R}^n$ is a continuous function that assigns each social state a payoff vector, one for each strategy in each population.
\end{definition}
$F^p_i : X \rightarrow \mathbb{R}$ is the payoff function for strategy $i \in S^p$.\\

$F^p : X \rightarrow \mathbb{R}^{n^p}$ denotes the payoff functions for all strategies in $S^p$.\\

\begin{definition} A game with payoffs defined on the positive orthant is called a complete population game, i.e. if one of the following two conditions is satisfied:\begin{itemize}
\item[(i)] $X=\lbrace x\in\mathbb{R}^n : \sum\limits_{k\in S} x_k=1\rbrace$
\item[(ii)] $X=\lbrace x=(x^1,...,x^\mathcal{P})\in\mathbb{R}^n_+: \sum\limits_{i\in S^p} x_i^p=m^p\rbrace$
\end{itemize}
\end{definition}

The procedures that agents follow to decide when to change policies and which policies to switch to are called review protocols.
\begin{definition} A revision protocol $\rho^p$ is an application $\rho^p: \mathbb{R}^{n^p} \times X^p \rightarrow \mathbb{R}_+^ {n^p\times n^p}$. The scalar $\rho^p_{ij}(\pi^p, x^p)$ is called the conditional passing rate from strategy $i \in S^p$ to strategy $j \in S^p$ given the gain vector $\pi^p$ and the population state $x^p$.
\end{definition}
We write $\rho^p_{ij}(\pi^p)$ instead of $\rho^p_{ij}(\pi^p, x^p)$ when the conditional switching rates do not directly depend of $x^p$.\\

\begin{definition}
The mean dynamic associated with the population game $F$ and the protocol $\rho$ is the ordinary differential equation:\begin{center}
$\dot{x}_i=\sum\limits_{j\in S}x_j\rho_{ji}-x_i\sum\limits_{j\in S}\rho_{ij}$
\end{center}
\end{definition}

We can therefore define our Markovian process more formally.
\begin{prop}
Let a population set $F$, a revision protocol $\rho$, a revision opportunity rate $R$, and a population size $N$ define a Markov process $(X^N)$ on the $\mathcal{X}^N$ state space. Let as initial state $X^N_0=x^N_0$ then its jump rate is $\lambda^N_x=NR$ (i.e. for each population $p$, $\lambda^{N^p}_x=N^pR ^p$) and its transition probabilities n are: \begin{center}
$P^{N}_{x,x+z}=\left\lbrace\begin{array}{l l}
\frac{x_i\rho_{ij}(F(x),x)}{R_{ij}}$ if $z=\frac{1}{N_{ij}}(e_j-e_i), \exists~ p $ such that $i,j\in S^p, i\neq j\\
1-\sum\limits_{i\in S}\sum\limits_{j\neq i}\frac{x_i\rho_{ij}(F(x),x)}{R_{ij}}$ if $z =0\\
0$ otherwise $
\end{array}\right.$
\end{center}
\end{prop}
$\forall i\in S$, $e_i\in\mathcal{X}^N$ is the state where all agents who can play strategy $i$.

\begin{prop} \textbf{(Deterministic approximation of $(X^N_t)$)} Let $(X^N_t)_{N\geq N_0}$ be the continuation of the evolutionary processes
stochastics defined above. Suppose that $V=V^N$ is Lipschitz and continuous. Let the initial conditions $X^N_0=x^N_0$ converge to the state $x_0\in X$, and let $(x_t)_{t\geq 0}$ be the solution of the average dynamics $(M)$ starting from $x_0$. So for all $T < \infty$ and $\varepsilon > 0$ we have:\begin{center}
$\lim\limits_{N\rightarrow\infty_{\mathbb{N}^\mathfrak{p}}}\mathbb{P}(\sup\limits_{t\in [0,T]}\vert X^N_t -x_t\vert <\varepsilon)=1$
\end{center}
\end{prop}

In finite horizon analyses, the basic object of study is the mean dynamic $(M)$, an ordinary differential equation derived from the Markov process $(X^N_t)$. In infinite horizon analyses, the corresponding object is the stationary distribution $\mu^N$ of the process $(X^N_t)$. A stationary distribution is defined by the property that a process whose initial condition is described by this distribution will continue to be described by this distribution at all future times.\\

We recall that $(X^N_t)$ is the Markov chain defined as in Proposition 0.1. To introduce the possibility of infinite horizon predictions, we now further assume that the conditional rates of change are strictly positive i.e. there is a constant $\underline{R}\in\mathbb{R} _+$ such that:
\begin{center}
~\hfill$\rho_{ij}(F(x), x) \geq \underline{R}~\forall~i, j \in S, x \in X$\hfill (1)
\end{center}
\begin{definition} A fully supported revision protocol is a protocol that satisfies condition (1).
\end{definition}
Full-support review protocols usually include some form of perturbation to ensure that all strategies are always chosen with positive probability. 
\section{The case one population}

In \cite{San10},  we are taught that the infinite horizon is known for only two types of population games: exponential games and two-strategy games, we will extend this to what are called symmetric population games
\subsection{Informal presentation with an example}

First, we will consider the case where there is only one population. \\

Imagine a population of $N$ agents playing a three-states game: state 1, state2 or state 3. This is a game with one population and three strategies, but it can be transformed into a game with three populations and two strategies. Imagine three populations of $N$ agents each (as if we were cloning the initial population), the first population would play a two-state game which would be: do strategy 1 or do not strategy 1, the second would play: do strategy 2 or do not strategy 2 and the third would play: do strategy 3 or do not strategy 3. The possible states would only be those where $N$ agents choose an "odd" strategy and $2N$ agents choose an "even" strategy. We have thus transformed the game into a two-strategy game that we know is reversible.\\
To switch from the initial game to the new game, you have to make small changes.\\

Let $A$ be an object concerning the initial game, its transformation for the new game will be $A^*$\\

In this case, \begin{itemize}
\item[•] $x=(x_1,x_2,x_3)\in X$ becomes $x^*=(x_1, 1-x_1,x_2,1-x_2,x_3,1-x_3)\in X^*$
\item[•] $S=\lbrace a, b, c\rbrace$ becomes $S^*=\lbrace a, \overline{a},b,\overline{b},c,\overline{c}\rbrace$
\item[•] $\mathfrak{p}=1$ becomes $\mathfrak{p}^*=3$
\item[•] $F(x)=(\alpha,\beta,\gamma)$ becomes $F^*(x^*)=(\alpha,\frac{x_2\beta+x_3\gamma}{x_2 +x_3},\beta,\frac{x_1\alpha+x_3\gamma}{x_1+x_3},\gamma,\frac{x_1\alpha+x_2\beta}{x_1+x_2})$ or \\ $F ^*(x^*)=(\alpha,0,\beta,0,\gamma,0)$
\item[•] $R$ becomes $R^*=(R,R,R)$
\item[•] \begin{align*}
\rho(\pi,x)&=\left(\begin{array}{c c c c c}
\rho_{a,a}(\pi,x) & \rho_{a,b}(\pi,x) & \rho_{a,c}(\pi,x) \\
\rho_{b,a}(\pi,x) & \rho_{b,b}(\pi,x) & \rho_{b,a}(\pi,x) \\
\rho_{c,a}(\pi,x) & \rho_{c,b}(\pi,x) & \rho_{c,c}(\pi,x) \\
\end{array}\right)
\end{align*}
Becomes
\begin{small}
\begin{align*}
\rho^*(\pi^*,x^*) = & \left(\begin{array}{c c c c c c}
\rho^*_{a,a}(\pi^*,x^*) & \rho^*_{a,\overline{a}}(\pi^*,x^*) & 0 & 0 & 0 & 0\\
\rho^*_{\overline{a},a}(\pi^*,x^*) & \rho^*_{\overline{a},\overline{a}}(\pi^*,x ^*) & 0 & 0 & 0 & 0\\
0&0&\rho^*_{b,b}(\pi^*,x^*) & \rho^*_{b,\overline{b}}(\pi^*,x^*) & 0 & 0 \\
0&0&\rho^*_{\overline{b},b}(\pi^*,x^*) & \rho^*_{\overline{b},\overline{b}}(\pi^*, x^*) & 0 & 0\\
0&0&0&0&\rho^*_{c,c}(\pi^*,x^*) & \rho^*_{c,\overline{c}}(\pi^*,x^*)\\
0&0&0&0&\rho^*_{\overline{c},c}(\pi^*,x^*) & \rho^*_{\overline{c},\overline{c}}(\pi^*, x^*)\\
\end{array}\right)
\
\\
= &\left(\begin{array}{c c}
\rho_{a,a}(\pi,x) & \rho_{a,b}(\pi,x) + \rho_{a,c}(\pi,x)\\
\frac{1}{2}(\rho_{b,a}(\pi,x) + \rho_{c,a}(\pi,x)) & \frac{1}{2}(\rho_{ b,b}(\pi,x)+\rho_{b,c}(\pi,x)+\rho_{c,b}(\pi,x)+\rho_{c,c}(\pi, x))\\
0&0\\
0&0\\
0&0\\
0&0 \\
\end{array}\right.\\
& \left.\begin{array}{c c}
0&0\\
0&0\\
\rho_{b,b}(\pi,x) & \rho_{b,a}(\pi,x) + \rho_{b,c}(\pi,x)\\
\frac{1}{2}(\rho_{a,b}(\pi,x) + \rho_{c,b}(\pi,x))& \frac{1}{2}(\rho_{ a,a}(\pi,x)+\rho_{a,c}(\pi,x)+\rho_{c,a}(\pi,x)+\rho_{c,c}(\pi, x)) \\
0&0\\
0&0\\
\end{array}\right.\\
& \left.\begin{array}{c c}
0&0 \\
0&0 \\
0&0 \\
0&0 \\
\rho_{c,c}(\pi,x) & \rho_{c,a}(\pi,x) + \rho_{c,b}(\pi,x) \\
\frac{1}{2}(\rho_{a,c}(\pi,x) + \rho_{b,c}(\pi,x)) & \frac{1}{2}(\rho_{a,a}(\pi,x)+\rho_{a,b}(\pi,x) +\rho_{b,a}(\pi,x)+\rho_{b,b}(\pi,x)) \\
\end{array}\right)
\end{align*}
\end{small}
\end{itemize}
We then know by theorem 11.2.3 of \cite{San10} what shape should have the stationary distribution of the game and therefore its infinite horizon.\\
\subsection{Formal presentation}

In this part we will lay the mathematical foundations to have access to the behavior when the time becomes very long for a game with $3$ strategies.\\
\begin{definition} $A$ is a symetric game if for all $i,j\in S$ we have $\rho_{ij}(\pi,x)=\rho_{ji}(\pi,x)$ and $R_{ij}=R_{ji}$.
\end{definition}
Let $A$ be a game of a population with $3$ strategies that we will call the initial game, in connection with this game we then pose:
\begin{itemize}
\item[•] $X$ the state space of $A$
\item[•] $S=\lbrace 1,2,3\rbrace$ the set of $3$ strategies in $A$
\item[•] We have $\mathfrak{p}=1$
\item[•] Let $F$ be its gain function
\item[•] We set $\rho$ its full support revision protocol
\item[•] $R$ its revision rate
\item[•] $N$ the number of agents in the population
\end{itemize}
\begin{prop} If for all $i,j\in S$ we have $\rho_{ij}(\pi,x)=\rho_{ji}(\pi,x)$ and $R_{ij}=R_{ji}$, then for any population game at $3$ strategies there is a unique population game with two strategies and $3$ population describing the same game.
\end{prop}
\begin{proof} To demonstrate this we must first build this new game, ie show its existence, and then show its uniqueness. \\

\textbf{Existence:} Let $A^*$ be a two-strategy game of $n$ population defined as follows: \begin{itemize}
\item[•] $X^*=\lbrace x^*=(x_1,1-x_1,x_2,1-x_2,x_3,1-x_3)$ such that $\exists~ x=(x_1,x_2,x_3 )\in X\rbrace$ its state space
\item[•] $S^*=\lbrace 1, \overline{1}, 2, \overline{2},3,\overline{3}\rbrace$ the set of $6 (2\times 3)$ strategies in $A^*$
\item[•] We have $\mathfrak{p}^*=3$ populations
\item[•] Let $x\in X$ such that $F(x)=y=(y_1,y_2,y_3)$ then\\ $F^*(x^*)=y^*=(y_1, 0,y_2,0,y_3,0)\in\mathbb{R}^{6}$
\item[•] Let $N^*=(N,N,N)\in(\mathbb{N}^*)^3$ be the number of agents in each population (which is $N$ for each)
\item[•] \begin{small}
\begin{align*}
\rho^*(\pi^*,x^*) = & \left(\begin{array}{c c c c c c}
\rho^*_{1,1}(\pi^*,x^*) & \rho^*_{1,\overline{1}}(\pi^*,x^*) & 0 & 0 & 0 & 0\\
\rho^*_{\overline{1},1}(\pi^*,x^*) & \rho^*_{\overline{1},\overline{1}}(\pi^*,x ^*) & 0 & 0 & 0 & 0\\
0&0&\rho^*_{2,2}(\pi^*,x^*) & \rho^*_{2,\overline{2}}(\pi^*,x^*) & 0 & 0 \\
0&0&\rho^*_{\overline{2},2}(\pi^*,x^*) & \rho^*_{\overline{2},\overline{2}}(\pi^*, x^*) & 0 & 0\\
0&0&0&0&\rho^*_{3,3}(\pi^*,x^*) & \rho^*_{3,\overline{3}}(\pi^*,x^*)\\
0&0&0&0&\rho^*_{\overline{3},3}(\pi^*,x^*) & \rho^*_{\overline{3},\overline{3}}(\pi^*, x^*)\\
\end{array}\right)\\
\end{align*}
\end{small}
with $\forall i\in S$ :
\begin{itemize}
\item[(i)] $\rho^*_{i,i}(\pi^*,x^*) = \rho_{i,i}(\pi,x)$
\item[(ii)] $\rho^*_{i,\overline{i}}(\pi^*,x^*) =\sum\limits_{\substack{j\in S \\ j\neq i}} \rho_{i,j}(\pi,x)$
\item[(iii)] $\rho^*_{\overline{i},i}(\pi^*,x^*) =\frac{1}{2}\sum\limits_{\substack{j \in S \\ j\neq i}} \rho_{j,i}(\pi,x)$
\item[(iv)] $\rho^*_{\overline{i},\overline{i}}(\pi^*,x^*) =\frac{1}{2}\left(\sum \limits_{\substack{j\in S \\ j\neq i}}\sum\limits_{\substack{k\in S\\ k\neq i \\ k\neq j}} \rho_{j,k }(\pi,x)+\sum\limits_{\substack{j\in S \\ j\neq i}}\rho_{j,j}(\pi,x)\right)$
\end{itemize}
\item[•] Let $R^*$ the new revision rate the same as $\rho^*$
\end{itemize}
~\\
\textbf{Unicity:} We now need to show that from each element previously constructed for $A^*$ we can find the element of $A$ used, in other words, that there is a bijection between each element of $A^*$ and the element of $A$ which allowed its construction.\\
\begin{itemize}
\item[•] To go from $N^*$ to $N$, just take the first coordinate of $N^*$
\item[•] To go from $\mathfrak{p}^*$ to $\mathfrak{p}$ just set $\mathfrak{p}=\frac{\mathfrak{p}^*}{3} $
\item[•] Let $f:\mathbb{R}^{6}\rightarrow\mathbb{R}^3$ such that $f(x_1,...,x_{6})=(x_1,x_3, x_5)$, then $f(X^*)=X$ and $f(F(X^*))=F(X)$
\item[•] By writing $S^*=\lbrace 1,...,6\rbrace$ we can find $S$ by setting $S=\lbrace 1,3,5\rbrace$
\item[•] Recall that $\forall i,j\in S, \rho_{ij}(\pi,x)=\rho_{ji}(\pi,x)$\begin{itemize}
\item[(i)] $\forall i\in S, \rho_{ii}(\pi,x)=\rho_{ii}^*(\pi^*,x^*)$
\item[(ii)] $\forall i,j,k\in S$ with $i\neq j\neq k$, we have:
\begin{align*}
\rho_{ij}(\pi,x)+\rho_{ji}(\pi,x) = & (\rho_{ij}(\pi,x)+\rho_{ik}(\pi,x)) +(\rho_{ji}(\pi,x)+\rho_{jk}(\pi,x))\\ & -(\rho_{ik}(\pi,x)+\rho_{jk}(\ pi,x))\\
= & \rho^*_{i,\overline{i}}(\pi^*,x^*) + \rho^*_{j,\overline{j}}(\pi^*,x^* ) -2 \rho^*_{\overline{k},k}(\pi^*,x^*)
\end{align*}
we then have $\forall i,j,k\in S$ with $i\neq j\neq k$ :\begin{center}
$\rho_{ij}(\pi,x)=\frac{\rho^*_{i,\overline{i}}(\pi^*,x^*) + \rho^*_{j,\overline{j}}(\pi^*,x^*) -2 \rho^*_{\overline{k},k}(\pi^*,x^*)}{2}$
\end{center}
\end{itemize}
\item[•] For $R^*$ we prove uniqueness as for the previous point.
\end{itemize}
\end{proof}
We will call $A^*$ the transformation of $A$.
\begin{prop} If for all $i,j\in S$ we have $\rho_{ij}(\pi,x)=\rho_{ji}(\pi,x)$, and $R_{ij}=R_{ji}$ then for any population game from $N$ agents to $3$ strategies with a full support revision protocol one can define a stationary distribution for the evolutionary process $(X^N_t)$ on $\mathcal{X}^N$.
\end{prop}

\begin{proof}

Let $A$ be a population game of $N$ agents with $3$ strategies with a satisfactory full-support revision protocol for all $i,j\in S$ that $\rho_{ij}(\pi,x)= \rho_{ji}(\pi,x)$. Then by applying Proposition 3.1, we can define the game $A^*$ which is the transformation of $A$. We can then apply theorem 11.2.3 of\cite{San10} which gives us the stationary distribution of the set $A^*$ which is for its population $i\in\llbracket 1,3\rrbracket$:  \begin{center}
$\frac{\mu^{N_i}_{\mathcal{X}_i}}{\mu^{N_i}_0}=\prod\limits^{N\mathcal{X}_i}_{j=1} \frac{N-j-1}{j}\cdot\frac{\rho^*_{i,\overline{i}}(F^{*i}(\frac{j-1}{N}),{ \frac{j-1}{N}})}{\rho^*_{\overline{i},i}(F^{*i}(\frac{j}{N}),{\frac{ j}{N}})}$ with $\mathcal{X}_i\in\lbrace 0,\frac{1}{N},\frac{2}{N},...,1\rbrace$ and $\sum\limits_{i=1}^3\mathcal{X}_i=1$
\end{center}
with $\mu^{N_i}_0$ determined by the requirement that $\sum\limits_{\mathcal{X}_i\in\mathcal{X}^{N_i}}\mu^{N_i}_{\mathcal {X}_i} = 1$.\\

\textbf{Remark :} Maybe have to make a mixture model to return to dimension 1

\end{proof}
\section{The case $\mathfrak{p}$ populations}
In this part we will show that proposition 1.2 is true for several populations, for this we will show a more general result by showing that we can find an infinite horizon that the number of strategy of the game of each population is 2 or 3 .\\

\begin{prop}
Let $A$ be a game with $\mathfrak{p}$ populations such that the population $p$ plays a game with $n^p$ strategies ($n^p=2$ or $3$), that its revision protocol is with full support and that for all $i,j\in S$ we have $\rho^p_{ij}(\pi,x)=\rho^p_{ji}(\pi,x)$ and $R_{ij}=R_{ji}$ then we can define a stationary distribution for the evolutionary process $(X^N_t)$ on $\mathcal{X}^N$.
\end{prop}
\begin{proof}
Let $p\in\llbracket 1,\mathfrak{p}\rrbracket~A^p$ be the game played by the population $p$. Applying Proposition 1.2, let $p\in\llbracket 1,\mathfrak{p}\rrbracket~A^{*p}$ be the transformation into a 2-strategy game of the game played by the population $p$ ($A^{ *p}=A$ if $n^p=2$). By setting $A^*$ the set of these games, we then end up with a game $A^*$ which is a 2-strategy game with $\sum\limits_{k=1}^\mathfrak{p} n^k=n$ populations, and $A^*$ represents the same game as $A$. We can then apply Theorem 11.2.3 of \cite{San10} which gives us the stationary distribution of the set $A^*$ which is for its population $i\in\llbracket 1,n\rrbracket$ such that $\exists p\in\llbracket 1,\mathfrak{p}\rrbracket$ such that $i\in\llbracket \sum\limits_{k=0}^{p-1}n^k+1,\sum\limits_{k =0}^{p}n^k\rrbracket$ (with $n^0=0)$: \begin{center}
$\frac{\mu^{N_i}_{\mathcal{X}_i}}{\mu^{N_i}_0}=\prod\limits^{N^p\mathcal{X}_i}_{j= 1}\frac{N^p-j-1}{j}\cdot\frac{\rho^*_{i,\overline{i}}(F^{*i}(\frac{j-1}{N ^p}),{\frac{j-1}{N^p}})}{\rho^*_{\overline{i},i}(F^{*i}(\frac{j}{ N^p}),{\frac{j}{N^p}})}$ with $\mathcal{X}_i\in\lbrace 0,\frac{1}{N^p},\frac{2 }{N^p},...,1\rbrace$ and $\sum\limits_{i=n^{p-1}+1}^{n^p}\mathcal{X}_i=1$ and $\sum\limits_{j=1}^{n}\mathcal{X}_j=\mathfrak{p}$
\end{center}
with $\mu^{N_i}_0$ determined by the requirement that $\sum\limits_{\mathcal{X}_i\in\mathcal{X}^{N_i}}\mu^{N_i}_{\mathcal {X}_i} = $1
\end{proof}
\section{Transformation of an n-strategy symetric game into a 2-strategy game}
To go from a game with $n$ strategies to a game with $2$ strategies, we will show that for a game with $n$ strategies there exists a unique game with $n-1$ strategies which represents the same game, then by repeating this we come across a game with $2$ strategies.\\

\begin{definition} $A$ is a symetric game if for all $i,j\in S$ we have $\rho_{ij}(\pi,x)=\rho_{ji}(\pi,x)$ and $R_{ij}=R_{ji}$.
\end{definition}

Let $A$ be a game of a population with $n$ strategies that we will call the initial game, in connection with this game we then pose:
\begin{itemize}
\item[•] $X$ the state space of $A$
\item[•] $S=\lbrace 1,2,...,n\rbrace$ the set of $n$ strategies in $A$
\item[•] We have $\mathfrak{p}=1$
\item[•] Let $F$ be its gain function
\item[•] We set $\rho$ its full support revision protocol
\item[•] $R$ its revision rate
\item[•] $N$ the number of agents in the population
\end{itemize}
\begin{prop} For any population game with $n$ strategies ($n>3$) there exists a unique population game with $n$ population each playing a game with $n-1$ strategies describing the same game .
\end{prop}
\begin{proof} To demonstrate this we must first build this new game, ie show its existence, and then show its uniqueness. \\

\textbf{Existence:} Let $A^*_{n-1}$ be a set of $n$ populations each having access to $n-1$ strategies defined as follows: \begin{itemize}
\item[•] $X^*_{n-1}=\lbrace x^*_{n-1}=(x_1,x_2,...,x_{n-2},1-x_{n- 1}-x_n,x_2,x_3,...,x_{n-1},1-x_1-x_n,...,x_n,x_1,x_2,...,x_{n-3},1-x_ {n-1}-x_{n-2})$ such that $\exists~ x=(x_1,...,x_n)\in X\rbrace$ its state space
\item[•] $S^*_{n-1}=\lbrace 1, 2,...,n-2,a_1,2,3,...,n-1,a_2,..., n,1,2,...,n-3,a_{n}\rbrace$ the set of $n-1\times n$ strategies in $A^*_{n-1}$
\item[•] We have $\mathfrak{p}^*_{n-1}=n$ populations
\item[•] Let $x\in X$ such that $F(x)=y=(y_1,...,y_n)$ then\\ $F^*_{n-1}(x^*_ {n-1})=y^*_{n-1}=(y_1,y_2,...,y_{n-2},0,y_2,y_3,...,y_{n-1}, 0,...,y_n,y_1,y_2,...,y_{n-3},0)$
\item[•] Let $N^*_{n-1}=(N,...,N)\in(\mathbb{N}^*)^n$ be the number of agents in each population ( which is $N$ for each)
\item[•] \begin{align*}
\rho^*_{n-1}(\pi^*_{n-1},x^*_{n-1}) = & \left(\begin{array}{c c c c c c}
A_1^{n-1} &0_{\mathcal{M}_{n-1}} & \cdots & 0_{\mathcal{M}_{n-1}}\\
0_{\mathcal{M}_{n-1}} & \ddots & \ddots & \vdots \\
\vdots&\ddots&\ddots & 0_{\mathcal{M}_{n-1}} \\
0_{\mathcal{M}_{n-1}}&\cdots&0_{\mathcal{M}_{n-1}}&A_n^{n-1}\\
\end{array}\right)\in\mathcal{M}_{n(n-1)} \\
\end{align*}
with $\forall i\in S$ :
\begin{align*} 
A_i^{n-1} = & \left(\begin{array}{c c c c c c}
\rho_{(i+0)_{\equiv n},(i+0)_{\equiv n}} &\rho_{(i+0)_{\equiv n},(i+1)_{\equiv n}} & \cdots & \rho_{(i+0)_{\equiv n},(i+n-3)_{\equiv n}}&\rho^*_{{n-1}_{(i+0)_{\equiv n},a_i}}\\
\rho_{(i+1)_{\equiv n},(i+0)_{\equiv n}} &\rho_{(i+1)_{\equiv n},(i+1)_{\equiv n}} & \cdots & \rho_{(i+1)_{\equiv n},(i+n-3)_{\equiv n}}&\rho^*_{{n-1}_{(i+1)_{\equiv n},a_i}}\\
\vdots & \vdots & \ddots & \vdots& \vdots \\
\rho_{(i+n-3)_{\equiv n}(,i+0)_{\equiv n}} &\rho_{(i+n-3)_{\equiv n},(i+1)_{\equiv n}} & \cdots & \rho_{(i+n-3)_{\equiv n},(i+n-3)_{\equiv n}}&\rho^*_{{n-1}_{(i+n-3)_{\equiv n},a_i}}\\
\rho^*_{{n-1}_{a_i,(i+0)_{\equiv n}}} &\rho^*_{{n-1}_{a_i,(i+1)_{\equiv n}}} & \cdots & \rho^*_{{n-1}_{a_i,(i+n-3)_{\equiv n}}}&\rho^*_{{n-1}_{a_i,a_i}}\\
\end{array}\right) \\
& \in\mathcal{M}_{n-1} \\
\end{align*}
with $\forall j\in S$: \begin{itemize}
\item[•] $\rho^*_{{n-1}_{j,a_i}}= \rho_{j,(i+n-2)_{\equiv n}}+\rho_{j, (i+n-1)_{\equiv n}}$
\item[•] $\rho^*_{{n-1}_{a_i,j}}=\rho_{(i+n-2)_{\equiv n},j}+\rho_{(i +n-1)_{\equiv n},j}$
\item[•] $\rho^*_{{n-1}_{a_i,a_i}}=\rho_{(i+n-2)_{\equiv n},(i+n-2)_{\equiv n}}+\rho_{(i+n-2)_{\equiv n},(i+n-1)_{\equiv n}}\rho_{(i+n-1)_{\equiv n},(i+n-2)_{\equiv n}}+\rho_{(i+n-1)_{\equiv n},(i+n-1)_{\equiv n}}$
\end{itemize}
where $\forall a\in\mathbb{N}~a_{\equiv n}$ is $a$ modulo $n$.
\item[•] We set $R^*_{n-1}$ the new revision rate in the same way as $\rho^*_{n-1}$
\end{itemize}
~\\
\textbf{Unicity:} We now need to show that from each element previously constructed for $A^*$ we can find the element of $A$ used, in other words, that there is a bijection between each element of $A^*$ and the element of $A$ which allowed its construction.\\
\begin{itemize}
\item[•] To go from $N^*_{n-1}$ to $N$, just take the first coordinate of $N^*_{n-1}$
\item[•] To go from $\mathfrak{p}^*_{n-1}$ to $\mathfrak{p}$ just set $\mathfrak{p}=\frac{\mathfrak{p} ^*_{n-1}}{n}$
\item[•] Let's set $f:\mathbb{R}^{n^2-n}\rightarrow\mathbb{R}^n$ such that $f(x_1,x_2...,x_{n^2-n})=(x_1,... ,x_{pn-p-1},...,x_{n^2-2n-1})$ with $p\in\llbracket 1, n-1\rrbracket$, then $f(X^*) =X$ and $f(F(X^*))=F(X)$
\item[•] By writing $S^*_{n_1}=\lbrace 1,...,n^2-n\rbrace$ we can find $S$ by setting $S=\lbrace 1,... ,pn-p-1,...,n^2-2n-1\rbrace$ with $p\in\llbracket 1, n-1\rrbracket$
\item[•] \begin{itemize}
\item[(i)] $\forall i\in S, \rho_{ii}(\pi,x)=\rho^*_{{n-1}_{(i-1)n-(i- 1)+1,(i-1)n-(i-1)+1}}(\pi^*_{n-1},x^*_{n-1})$ (this is the first coordinate of $A^{n-1}_i(\pi^*_{n-1},x^*_{n-1})$)\\
\item[(ii)] $\forall i\in S$ and $j\in\llbracket (i+1)_{\equiv n}, (i+n-3)_{\equiv n}\rrbracket$ such that $(i+a)_{\equiv n}=j$, $~\rho_{ij}(\pi,x)=\rho^*_{{n-1}_{(i-1) n-(i-1)+1,(i-1)n-(i-1)+1+a}}(\pi^*_{n-1},x^*_{n-1}) $
\item[(iii)] $\forall i\in S$ and $j=(i+n-1)_{\equiv n}$, $\rho_{ij}(\pi,x)=\rho^ *_{{n-1}_{(j-1)n-(j-1)+2,(j-1)n-(j-1)+1}}(\pi^*_{n- 1},x^*_{n-1})$
\item[(iv)] $\forall i\in S$ and $j=(i+n-2)_{\equiv n}$, \begin{center}
$\rho_{ij}(\pi,x)=\rho^*_{{n-1}_{(i-1)n-(i-1)+1,(i-1)n-(i -1)+n-1}}(\pi^*_{n-1},x^*_{n-1})$\\
\textcolor{white}{$\rho_{ij}(\pi,x)=$} $- \rho^*_{{n-1}_{(j-1)n-(j-1)+2 ,(j-1)n-(j-1)+1}}(\pi^*_{n-1},x^*_{n-1})$
\end{center}
\end{itemize}
\item[•] For $R^*_{n-1}$ we prove uniqueness as for the previous point.
\end{itemize}
\end{proof}
Now that we know how to go from a game with $n$ strategies to a game with $n-1$ strategies, we can explain how to go from a game with $n$ strategy to a game with $m$ strategy with $m \leq n$.
\begin{prop} For any population game with $n$ strategies ($n>3$) there exists a population game with $n$ population each playing a game with $m$ strategies ($m\leq n$ ) describing the same game.
\end{prop}
\begin{proof}
Let $A$ be the initial game with $n$ strategies, let $A^*_{n-1}$ be its transformation into a game with $n-1$ strategies given by Proposition 1.1. Applying proposition 1.1 again to the game $A^*_{n-1}$ we find a game $B^*_{n-1}$ defined as follows:
\begin{itemize}
\item[•] $X^*_{B^*_{n-1}}=\lbrace x^*_{B^*_{n-1}}=(x_1,x_2,...,x_ {n-3},1-x_{n-2}-x_{n-1}-x_n,...,x_n,x_1,x_2,...,x_{n-4},1-x_{n -1}-x_{n-2}-x_{n-3})\in\mathbb{R}^{n\times (n-1)\times (n-2)}$ such that $\exists~ x=(x_1,...,x_n)\in X\rbrace$ its state space
\item[•] $S^*_{B^*_{n-1}}=\lbrace 1, 2,...,n-3,a_1,2,3,...,n-2, a_2,...,n,1,2,...,n-4,a_{n}\rbrace$ the set of $(n-2)\times (n-1)\times n$ strategies in $A^*_{n-1}$
\item[•] We have $\mathfrak{p}^*_{B^*_{n-1}}=n\times (n-1)$ populations
\item[•] Let $x\in X$ such that $F(x)=y=(y_1,...,y_n)$ then\\ $F^*_{B^*_{n-1} }(x^*_{B^*_{n-1}})=y^*_{B^*_{n-1}}=(y_1,y_2,...,y_{n-3} ,0,y_2,y_3,...,y_{n-2},0,...,y_n,y_1,y_2,...,y_{n-4},0)$
\item[•] We set $N^*_{B^*_{n-1}}=(N,...,N)\in(\mathbb{N}^*)^{n\times n -1}$ the number of agents in each population (which is $N$ for each)
\item[•] \begin{align*}
\rho^*_{B^*_{n-1}}(\pi^*_{n-1},x^*_{n-1}) = & \left(\begin{array}{c c c c c c }
B_1^{n-1} &0_{\mathcal{M}_{n-2}} & \cdots & 0_{\mathcal{M}_{n-2}}\\
0_{\mathcal{M}_{n-2}} & \ddots & \ddots & \vdots \\
\vdots&\ddots&\ddots & 0_{\mathcal{M}_{n-2}} \\
0_{\mathcal{M}_{n-2}}&\cdots&0_{\mathcal{M}_{n-2}}&B_{n\times (n-1)}^{n-1}\\
\end{array}\right)\in\mathcal{M}_{n(n-1)(n-2)} \\
\end{align*}
with $\forall i\in S$ :
\begin{align*} 
B_i^{n-1} = & \left(\begin{array}{c c c c c c}
\rho_{(i+0)_{\equiv n},(i+0)_{\equiv n}} &\rho_{(i+0)_{\equiv n},(i+1)_{\equiv n}} & \cdots & \rho_{(i+0)_{\equiv n},(i+n-4)_{\equiv n}}&\rho^*_{{B^*_{n-1}}_{(i+0)_{\equiv n},a_i}}\\
\rho_{(i+1)_{\equiv n},(i+0)_{\equiv n}} &\rho_{(i+1)_{\equiv n},(i+1)_{\equiv n}} & \cdots & \rho_{(i+1)_{\equiv n},(i+n-4)_{\equiv n}}&\rho^*_{{B^*_{n-1}}_{(i+1)_{\equiv n},a_i}}\\
\vdots & \vdots & \ddots & \vdots& \vdots \\
\rho_{(i+n-4)_{\equiv n}(,i+0)_{\equiv n}} &\rho_{(i+n-4)_{\equiv n},(i+1)_{\equiv n}} & \cdots & \rho_{(i+n-4)_{\equiv n},(i+n-4)_{\equiv n}}&\rho^*_{{B^*_{n-1}}_{(i+n-4)_{\equiv n},a_i}}\\
\rho^*_{{B^*_{n-1}}_{a_i,(i+0)_{\equiv n}}} &\rho^*_{{B^*_{n-1}}_{a_i,(i+1)_{\equiv n}}} & \cdots & \rho^*_{{B^*_{n-1}}_{a_i,(i+n-4)_{\equiv n}}}&\rho^*_{{B^*_{n-1}}_{a_i,a_i}}\\
\end{array}\right) \\
& \in\mathcal{M}_{n-1} \\
\end{align*}
\item[•] Let $R^*_{B^*_{n-1}}=(R,...,R)$ be the revision rate in each population (which is $R$ for each)
\end{itemize}
Let us then set $A^*_{n-2}$ a restriction of $B^*_{n-1}$ such that its revision protocol is: \begin{align*}
~ & \left(\begin{array}{c c c c c c c c}
B_1^{n-1} &0_{\mathcal{M}_{n-2}} & \cdots & 0_{\mathcal{M}_{n-2}}& \cdots & \cdots & 0_{\mathcal {M}_{n-2}}\\
0_{\mathcal{M}_{n-2}} & \ddots & \ddots & \vdots& \ddots & \ddots & \vdots \\
\vdots&\ddots&\ddots & 0_{\mathcal{M}_{n-2}} & \cdots & \cdots & 0_{\mathcal{M}_{n-2}}\\
0_{\mathcal{M}_{n-2}}&\cdots &0_{\mathcal{M}_{n-2}}&B_{in-i-1}^{n-1}&0_{\mathcal{ M}_{n-2}} & \cdots & 0_{\mathcal{M}_{n-2}}\\
0_{\mathcal{M}_{n-2}} & \cdots & \cdots & 0_{\mathcal{M}_{n-2}}& \cdots & \cdots & 0_{\mathcal{M}_ {n-2}}\\
0_{\mathcal{M}_{n-2}} & \ddots & \ddots & \vdots& \ddots & \ddots & \vdots \\
\vdots&\ddots&\ddots & \vdots & \cdots & \cdots & 0_{\mathcal{M}_{n-2}}\\
0_{\mathcal{M}_{n-2}}&\cdots&0_{\mathcal{M}_{n-2}}& 0_{\mathcal{M}_{n-2}} &0_{\mathcal{ M}_{n-2}} & \cdots & B_{n^2-n-1}^{n-1}\\
\end{array}\right)\in\mathcal{M}_{n(n-2)}\\
\end{align*}
We then say that $A^*_{n-2}$ is the transformation of $A$ into a game with $n$ populations playing a game with $n-2$ strategies, similarly, by iterating we can find $A^*_m$.
\end{proof}
\begin{prop} If for all $i,j\in S$ we have $\rho_{ij}(\pi,x)=\rho_{ji}(\pi,x)$ and $R_{ij}= R_{ji}$, then for any $n$ strategy population game with a fully supported revision protocol one can define a stationary distribution for the evolutionary process $(X^N_t)$ on $\mathcal{X}^ N$.
\end{prop}
\begin{proof}
Let $A$ be a symetric population game of $N$ agents with $n$ strategies with a fully supported revision protocol satisfying for all $i,j\in S$ that $\rho_{ij}(\pi,x) =\rho_{ji}(\pi,x)$. Then by applying propositions 4.1 and 4.2, we can define the game $A^*_3$ which is the transformation of $A$ into a game with $n$ populations playing a game with $3$ strategies. Then by applying Proposition 3.1, we can define the game $A^*$ which is the transformation of $A$. We can then apply theorem 11.2.3 of \cite{San} which gives us the stationary distribution of the set $A^*$ which is for its population $i\in\llbracket 1,n\rrbracket$: \begin{center}
$\frac{\mu^{N_i}_{\mathcal{X}_i}}{\mu^{N_i}_0}=\prod\limits^{N\mathcal{X}_i}_{j=1} \frac{N-j-1}{j}\cdot\frac{\rho^*_{i,\overline{i}}(F^{*i}(\frac{j-1}{N}),{ \frac{j-1}{N}})}{\rho^*_{\overline{i},i}(F^{*i}(\frac{j}{N}),{\frac{ j}{N}})}$ with $\mathcal{X}_i\in\lbrace 0,\frac{1}{N},\frac{2}{N},...,1\rbrace$ and $\sum\limits_{i=1}^n\mathcal{X}_i=1$
\end{center}
with $\mu^{N_i}_0$ determined by the requirement that $\sum\limits_{\mathcal{X}_i\in\mathcal{X}^{N_i}}\mu^{N_i}_{\mathcal {X}_i} = $1
\end{proof}
\textbf{Remark:} The chosen restriction is natural because it mainly consists of removing coordinates that appear several times.
\section{Symetric game with $\mathfrak{p}$ populations}
In this part we will show that proposition 3.3 is true for several populations, for this we will show a more general result by showing that we can find an infinite horizon whatever the number of strategies in the game of each population.\ \

\begin{prop}
Let $A$ be a symetric game with $\mathfrak{p}$ populations such that the population $p$ plays a game with $n^p$ strategies, its revision protocol has complete support and for all $i, j\in S$ we have $\rho^p_{ij}(\pi,x)=\rho^p_{ji}(\pi,x)$ and $R_{ij}=R_{ji}$ then we can define a stationary distribution for the evolutionary process $(X^N_t)$ on $\mathcal{X}^N$.
\end{prop}
\begin{proof}
Let $p\in\llbracket 1,\mathfrak{p}\rrbracket~A^p$ be the game played by the population $p$. Applying Proposition 3.3, let $p\in\llbracket 1,\mathfrak{p}\rrbracket~A^{*p}$ be the transformation into a 2-strategy game of the game played by the population $p$. By setting $A^*$ the set of these games, we then end up with a game $A^*$ which is a 2-strategy game with $\sum\limits_{k=1}^\mathfrak{p} n^k=n$ populations, and $A^*$ represents the same game as $A$. We can then apply Theorem 11.2.3 of \textbf{[San10]} which gives us the stationary distribution of the set $A^*$ which is for its population $i\in\llbracket 1,n\rrbracket$ such that $ \exists p\in\llbracket 1,\mathfrak{p}\rrbracket$ such that $i\in\llbracket \sum\limits_{k=0}^{p-1}n^k+1,\sum\limits_ {k=0}^{p}n^k\rrbracket$ (with $n^0=0)$: \begin{center}
$\frac{\mu^{N_i}_{\mathcal{X}_i}}{\mu^{N_i}_0}=\prod\limits^{N^p\mathcal{X}_i}_{j= 1}\frac{N^p-j-1}{j}\cdot\frac{\rho^*_{i,\overline{i}}(F^{*i}(\frac{j-1}{N ^p}),{\frac{j-1}{N^p}})}{\rho^*_{\overline{i},i}(F^{*i}(\frac{j}{ N^p}),{\frac{j}{N^p}})}$ with $\mathcal{X}_i\in\lbrace 0,\frac{1}{N^p},\frac{2 }{N^p},...,1\rbrace$ and $\sum\limits_{i=n^{p-1}+1}^{n^p}\mathcal{X}_i=1$ and $\sum\limits_{j=1}^{n}\mathcal{X}_j=\mathfrak{p}$
\end{center}
with $\mu^{N_i}_0$ determined by the requirement that $\sum\limits_{\mathcal{X}_i\in\mathcal{X}^{N_i}}\mu^{N_i}_{\mathcal {X}_i} = $1
\end{proof}
~~\\

\textbf{Remark:} Not all the properties of this article are subject to the symmetry condition of the protocol and the revision rate, so if the infinite horizon of 3-strategy games is found by applying the same properties, we will have the infinite horizons of all games on the condition of assuming the inertia and myopia of the players.

\section*{Acknowledgement}
I would like to thank Professor Jean-Réné Chazottes for introducing me to the dynamics of evolutionary play, which is a fascinating subject, as well as for the advice he gave me to carry out this work.
I would also like to thank Théophile Truchis for the time we were able to work in threes on the subject.

\end{document}